\makeatletter
\def\input@path{{styles/}}
\makeatother
\documentclass{ifacconf}
\usepackage{natbib}
\usepackage{url}
\usepackage{eurosym}
\usepackage{tabularx}
\usepackage{longtable}
\usepackage{float}
\usepackage{placeins}
\usepackage[table,dvipsnames]{xcolor}
\usepackage{graphicx}
\usepackage{algorithm}
\usepackage{algpseudocode}
\usepackage{array}
\usepackage{multirow}
\usepackage{textcomp}
\usepackage{stfloats}
\usepackage{verbatim}
\usepackage{amsmath,amssymb,amsfonts,mathtools,bm}
\usepackage{tikz}
\usepackage{siunitx}
\usepackage{arydshln}
\usepackage{balance}
\usepackage{booktabs}
\usepackage{bbold}
\usepackage{marginnote}
\usepackage{soul}
\usepackage[makeroom]{cancel}
\usepackage{cases}
\usepackage[hidelinks,implicit=false]{hyperref}

\newtheorem{assumption}{Assumption}

\newtheorem{remark}{Remark}

\usepackage{eso-pic}  % for arxiv statement

\makeatletter
\renewcommand\appendix{
    \par
    \setcounter{section}{0}
    \setcounter{subsection}{0}
    \gdef\thesection{\Alph{section}}
}
\makeatother

\begin{document}

% arxiv statement
\AddToShipoutPictureFG*{%
  \AtPageUpperLeft{%
    \raisebox{-25mm}{%
      \makebox[\paperwidth][c]{%
        \footnotesize
        This work has been submitted to IFAC for possible publication.
      }%
    }%
  }%
}

\begin{frontmatter}

% IFAC: 10 words at maximum

% Scalable
% Partial Measurements
% Structure-Exploiting
% Nonlinear
% Physical and Neural Dynamics

% Scalable Identification of Physical-Neural Dynamics from Partial Measurements: A Structure-Exploiting PEM–MHE Approach

%\title{Scalable Identification of general Parametric Models from noisy Measurements}
\title{Scalable Identification of\\Physics-Based and Neural Dynamics\\from Partial Measurements}

% \thanks[footnoteinfo]{...}

\author[Chalmers,Imperial]{Yibo Zhou}
\author[Imperial]{Johannes K\"ohler}

\address[Chalmers]{Department of Electrical Engineering, Chalmers University of Technology, Gothenburg, Sweden (e-mail: yiboz@chalmers.se).}
\address[Imperial]{Department of Mechanical Engineering, Imperial College London, London, United Kingdom (e-mail: j.kohler@imperial.ac.uk).}

\begin{abstract} % Abstract of 50-100 words
We present a scalable method for identifying nonlinear dynamics from partial measurements, supporting physics-based, neural-network, and gray-box models.
We use the Prediction Error Method (PEM) coupled with state estimation performed by a Moving Horizon Estimation (MHE) scheme, which can account directly for general nonlinear dynamics. This results in a bi-level optimization problem that scales poorly with increasing model parameters. 
To improve computational efficiency, we derive a reformulation that avoids explicit parameter-wide sensitivity matrices and exploits the temporal structure of the underlying optimization problem.
Numerical comparisons show up to a $215\times$ speedup in hypergradient evaluation.
Racing-vehicle simulations demonstrate accurate prediction and physical-parameter recovery in tire-force identification.
Real-world data show that the gray-box model improves prediction accuracy with a neural residual while retaining physical interpretability.
\end{abstract}

%\begin{keyword}Nonlinear system identification; Grey-box modelling; Parameter and state estimation; Parametric sensitivity analysis; Large scale optimization problems; Vehicle dynamic systems;\end{keyword}

\end{frontmatter}

% Submission notes:
% Special issue: https://www.sciencedirect.com/special-issue/331949/recent-topics-in-system-identification-learning-models-for-decision-and-control-special-issue-from-ifac-sysid-2027
% Joint journal/conference submission deadline: October 11, 2026.
% IFAC Symposia format; initial submissions must be 4-8 pages.

\section{Introduction}
\label{sec:introduction}

Accurate dynamical models are essential for prediction and control design. In this context, system identification provides established methods for learning models directly from measured data~\citep{schoukens2019nonlinear}. Approaches based on first principles exploit known physical structure and retain interpretable parameters. However, some dynamical effects remain difficult to model analytically, motivating flexible black-box representations such as neural networks (NNs)~\citep{chen2018neural}. Gray-box models combine these advantages by retaining physical structure and using neural networks to represent unknown dynamics~\citep{karniadakis2021physics, forgione2021continuous}. Identifying such gray-box models can be challenging when full-state measurements are unavailable, since state reconstruction must be coupled with model identification. In this paper, we focus on identifying gray-box models using noisy partial state measurements and high-dimensional neural network parameterizations.

\textit{Neural state-space models:} 
In the case of full-state measurements, continuous-time dynamics can be learned using neural ODEs by minimizing simulation errors and differentiating through numerical integration \citep{chen2018neural}. Alternatively, neural dynamics can be learned using standard regression for discrete-time systems~\citep{harrison2020meta}.

When the states are not fully measured, past measured inputs and outputs can be used directly to represent an extended state and learn a neural NARX model~\citep{lin1996learning}. Alternatively, recurrent models such as Long Short-term Memory (LSTM)~\citep{hochreiter1997long} propagate information recursively over time through latent states. Encoder-based methods, such as SUBNET~\citep{beintema2023deep} learn a mapping from past input-output data to latent states. Similar ideas are used to learn latent world models~\citep{ha2018recurrent,hafner2023mastering}, which also account for rich high-dimensional observations, such as images. 

Overall, these approaches can learn flexible neural dynamics from data, but provide limited options for directly incorporating prior physical knowledge because their internal representations need not correspond to physical quantities. However, for many applications, it is desirable to benefit from such prior knowledge while learning more flexible neural dynamics.

\textit{Physics-informed machine learning:}
To incorporate prescribed physical quantities, several approaches introduce prior physical knowledge into neural models. When the full state is measured, representative approaches include Physics-Informed Neural Networks (PINNs), which impose governing-equation residuals in the training objective \citep{raissi2019physics}, and Hamiltonian neural networks, which encode specific physical structure directly into the learned dynamics \citep{greydanus2019hamiltonian}.

When the state is only partially observed, the coupling between state and parameter estimation can be efficiently handled using Expectation Maximization (EM)~\citep{gibson2005robust} in the linear case. However, EM becomes computationally challenging for general nonlinear models~\citep{schon2011system}. To identify nonlinear dynamics from partial measurements, representative approaches include joint optimization of hidden-state trajectories and neural model parameters~\citep{forgione2021continuous}, and variational inference of unmeasured physical states alongside learning Gaussian-process dynamics~\citep{curi2020structured}.

 % equation discovery from learned latent coordinates~\citep{champion2019data} or reconstruction physical states~\citep{chen2021physics}.

% For the setting considered here, two properties are particularly desirable: predictive performance and preservation of the physical interpretation of the unmeasured states. The PEM--MHE framework of \cite{simpson2024parallelizable} provides a natural formulation combining these properties.
We are particularly interested in approaches that achieve accurate predictions while preserving the physical interpretation of unmeasured states. Our proposed approach builds on the PEM--MHE framework of \citet{simpson2024parallelizable}.
Prediction Error Method (PEM) is a classical system-identification method that estimates model parameters by minimizing the prediction error of an observer~\citep{ljung1998system}, which is well suited for systems with noisy partial measurement. Moving Horizon Estimation (MHE) is an optimization-based state estimator that can be applied to general nonlinear systems~\citep{rawlings2020model}. By embedding a fixed-window MHE within PEM, this approach estimates the model parameters through optimizing the general MHE observer. However, the framework was demonstrated only on small-scale systems with two unknown parameters. Extending it to complex neural models with a large number of parameters is computationally challenging, due to the nested optimization between model learning and state estimation.

\textit{Contribution:}
To address this scalability limitation, we develop a scalable extension of the PEM--MHE framework.
We focus on modeling generality, computational scalability, numerical implementation, and practical applicability of the framework, while statistical guarantees, such as asymptotic consistency, are beyond the scope of this work. Specifically, the proposed framework provides the following capabilities:

\begin{itemize}
    \item \textbf{Modeling flexibility:}
    A unified continuous-time formulation enables the identification of physical parameters and neural-network weights, supporting physics-based, neural, and gray-box models under partial state measurements.

    \item \textbf{Scalability with parameter number:}
    We use Vector--Jacobian Products (VJPs) to avoid forming parameter-wide derivative matrices. This improves computational and memory scaling with neural-network size.

    \item \textbf{Parallel efficiency:}
    A batched Karush-Kuhn-Tucker (KKT) implementation exploits independence across data windows and temporal sparsity within each window. This reduces training time and memory usage.

    \item \textbf{Numerical comparisons:}
    Numerical comparisons demonstrate up to a $215\times$ speedup in hypergradient evaluation.
    In simulation, the proposed approach achieves over $80\times$ lower validation prediction error than the best-performing baseline, while accurately recovering physical parameters.
    On real-world vehicle data, the gray-box model reduces prediction error by approximately $31\%$ compared with its physics-only counterpart.
\end{itemize}

\section{Problem Setup and Preliminaries}
\label{sec:problem-setup}

\subsubsection*{Notation:}
For a differentiable scalar-valued function $f:\mathbb{R}^n\rightarrow\mathbb{R}$, we denote its gradient with respect to $x$ by $\nabla_x f:=\left(\frac{\partial f(x)}{\partial x}\right)^\top\in\mathbb R^n$. By $Q\succ0$, we indicate that a matrix $Q$ is symmetric positive definite. For $Q\in\mathbb{R}^{n\times n}$ with $Q\succ0$, we denote the weighted Euclidean norm by $\|x\|_Q:=\sqrt{x^\top Qx}$. For vectors $x_1,\ldots,x_p$, we denote their vertical concatenation by $\operatorname{col}(x_1,\ldots,x_p):=[x_1^\top,\ldots,x_p^\top]^\top$.

\subsection{Setup}
\label{subsec:system-identification-setting}
We have a parametrized nonlinear dynamics model of the form 
\begin{subequations}
\label{eq:discrete-time-model}
\begin{align}
    x_{k+1} &= f_{\mathrm d}(x_k,u_k,w_k;\theta),
    \label{eq:discrete-time-dynamics} \\
    y_k &= Cx_k+F\eta_k.
    \label{eq:discrete-time-output}
\end{align}
\end{subequations}
with state $x_k\in\mathbb{R}^{n_{\mathrm{x}}}$, control input $u_k\in\mathbb{R}^{n_{\mathrm{u}}}$, measured output $y_k\in\mathbb{R}^{n_{\mathrm{y}}}$, process noise $w_k\in\mathbb{R}^{n_{\mathrm{w}}}$, measurement noise $\eta_k\in\mathbb{R}^{n_\eta}$, and time $k\in\mathbb{N}$.
The process noise sequence $\{w_k\}$ and measurement noise sequence $\{\eta_k\}$ are mutually independent, zero-mean, and independently and identically distributed over time (i.i.d.). 
Only the applied control inputs $u_k$ and measured outputs $y_k$ are available, whereas the latent state  $x_k$ is not known. 
The identification task is to determine the unknown model parameter $\theta\in\mathbb{R}^{n_\theta}$ from the input--output data. 

\begin{remark}[Model generality]
The transition map $f_{\mathrm d}$ may be obtained by integrating a continuous-time model
\begin{equation}
\label{eq:continuous-time-model}
        \dot{x}_{\tau}=f_{\mathrm{c}}(x_{\tau},u_{\tau},w_{\mathrm{c},\tau};\theta),
\end{equation}
where $f_{\mathrm{c}}$ describes the continuous-time dynamics and $w_{\mathrm{c}}$ corresponds to the continuous-time noise. Depending on the discretization, the discrete-time noise $w_k\in\mathbb{R}^{n_{\mathrm{w}}}$ accounts for $w_{\mathrm{c}}$ at multiple integration stages. Such a continuous-time formulation facilitates the incorporation of prior physical knowledge. The dynamics, represented either directly by $f_{\mathrm{d}}$ or through $f_{\mathrm{c}}$, may be parameterized as a physical model, a neural-network model, or a hybrid gray-box model. Accordingly, $\theta$ contains uncertain physical parameters in the purely physical case, neural-network parameters in the neural case, and both in the gray-box case. Neural and gray-box parameterizations can lead to large parameter dimensions, with $n_\theta\gg 1$, motivating computational methods that scale efficiently with $n_\theta$.
\end{remark}

We consider the following assumptions for the model.

\begin{assumption}[Known model structure and regularity]
The parametric structure of the dynamics $f_{\mathrm d}(\cdot;\theta)$ and the matrices $C\in\mathbb{R}^{n_{\mathrm{y}}\times n_{\mathrm{x}}}$ and $F\in\mathbb{R}^{n_{\mathrm{y}}\times n_\eta}$ are  known.
The dynamics $f_{\mathrm d}$ are twice continuously
differentiable.
\end{assumption}

\begin{assumption}[Final-state observability]
\label{ass:observable}
The model $(f_{\mathrm d},C)$ is final-state observable~\cite[Def.~ 4.29]{rawlings2020model} in $n_{\mathrm{o}}$ steps for all parameters $\theta\in\mathbb{R}^{n_{\mathrm{\theta}}}$ uniformly for all initial states and input trajectories. 
\end{assumption}

Assumption~\ref{ass:observable} implies that in the absence of noise and for known $\theta$, the current state is uniquely determined by the preceding $n_{\mathrm{o}}$ input and output measurements.

The available input--output data may be collected from one or multiple trajectories. The data are segmented into fixed-length, possibly overlapping windows, each containing $N$ consecutive inputs and the corresponding one-step-ahead outputs. The window length is chosen such that $N\geq n_{\mathrm{o}}$, allowing the terminal state to be determined under Assumption~\ref{ass:observable}. The resulting dataset consists of $M$ windows:
\begin{equation}
    \mathcal D
    =
    \left\{
    \left(
    \tilde u_{0:N-1}^{(i)},
    \tilde y_{1:N}^{(i)}
    \right)
    \right\}_{i=1}^{M}.
    \label{eq:windowed-dataset}
\end{equation}

\subsection{Preliminaries - Moving Horizon Estimation}
\label{subsec:mhe-preliminaries}

MHE is an optimization-based method for state estimation. In particular, given a known model $f_{\mathrm d}$, including known model parameters $\theta$, MHE uses the known past input sequence $\tilde u_{0:N-1}$ and output sequence $\tilde y_{1:N-1}$ to estimate the state trajectory $\hat x_{0:N}$ by solving the following optimization problem:
\begin{subequations}
\label{eq:inner-mhe}
\begin{equation}
    \begin{gathered}
    \min_{\hat{x},\hat{w},\hat{\eta}}
        % \mathcal{M}:=
        \left\lVert\hat x_0-\bar s\right\rVert_{\Sigma^{-1}}^2
        +\sum_{k=0}^{N-1}\left\lVert\hat w_k\right\rVert_{Q^{-1}}^2+\sum_{k=1}^{N-1}\left\lVert\hat\eta_k\right\rVert_{R^{-1}}^2
    \end{gathered}
    \label{eq:mhe-objective}
\end{equation}
\begin{align}
    &\text{s.t.}\notag\\[-0.2em]
    % \mathcal H_k^w:=
    &\hat x_{k+1}-f_{\mathrm d}(\hat x_k,\tilde u_k,\hat w_k;\theta)=0,
    \quad k=0,\ldots,N-1,
    \label{eq:mhe-dynamics-constraint}\\
    % \mathcal H_k^\eta:=
    &\tilde y_k-C\hat x_k-F\hat\eta_k=0,
    \quad k=1,\ldots,N-1.
    \label{eq:mhe-measurement-constraint}
\end{align}
\end{subequations}
with weighting matrices $Q\in\mathbb{R}^{n_{\mathrm{w}}\times n_{\mathrm{w}}}$, $R\in\mathbb{R}^{n_\eta\times n_\eta}$, and $\Sigma\in\mathbb{R}^{n_{\mathrm{x}}\times n_{\mathrm{x}}}$, where $Q,R,\Sigma\succ0$. This problem computes a trajectory that is consistent with the data and the model, using a quadratic objective that penalizes the estimated noise. This cost is related to a maximum-likelihood trajectory estimate under Gaussian noise \cite[Sec.~4]{simpson2023efficient}. The constant arrival-cost term $\left\lVert\hat x_0-\bar s\right\rVert_{\Sigma^{-1}}^2$, with  $\bar{s}\in\mathbb{R}^{n_{\mathrm{x}}}$, is not required for standard MHE, but plays an important role in the PEM--MHE framework to avoid biased parameter estimates \cite[Sec.~V.A]{simpson2024parallelizable}.

\begin{remark}
In contrast to classical recursive state estimators, such as the Kalman filter and extended Kalman filter (EKF), the fixed-window MHE formulation yields $M$ independent state-estimation problems. This independence allows the data windows to be processed separately, enabling both state estimation and subsequent gradient computations for parameter estimation to be performed in parallel \citep{simpson2024parallelizable}.
\end{remark}
\section{Method}
\label{sec:method}

In this section, we present the proposed methodology for estimating the unknown model parameters $\theta$ from the available dataset $\mathcal D$. First, Section~\ref{subsec:pem-mhe} presents the PEM--MHE identification framework based on~\cite{simpson2024parallelizable}. Building on this framework, Section~\ref{subsec:matrix-free-hypergradient} proposes a scalable, structure-exploiting reformulation for hypergradient computation, together with its theoretical properties.

\subsection{PEM-MHE Identification Framework}
\label{subsec:pem-mhe}

We follow the standard PEM approach of estimating the trainable parameters by minimizing prediction errors. These predictions are constructed using state estimates obtained from MHE, since the state is not directly measured.

For the outer PEM problem, let $\alpha\in\mathbb{R}^{n_\alpha}$ collect all parameters to be optimized during identification. In our formulation, $\alpha$ includes the model parameters $\theta$ and the arrival-cost parameters $\bar s$ and $\Sigma$. For the inner MHE problem, let $\mathrm v^{(i)}$ collect the decision variables $\hat x_{0:N}^{(i)}$, $\hat w_{0:N-1}^{(i)}$, and $\hat\eta_{1:N-1}^{(i)}$ for each data window $i$. We denote the corresponding objective \eqref{eq:mhe-objective}, dynamics constraints \eqref{eq:mhe-dynamics-constraint}, and measurement constraints \eqref{eq:mhe-measurement-constraint} by $\mathcal M^{(i)}(\mathrm v^{(i)},\alpha)$, $\mathcal H_k^{w,(i)}(\mathrm v^{(i)},\alpha)$, and $\mathcal H_k^{\eta,(i)}(\mathrm v^{(i)},\alpha)$, respectively.

With this notation, the PEM--MHE bilevel optimization problem is formulated as follows:
\begin{subequations}
\label{eq:pem-objective}
\begin{align}
    \alpha^*
    &\in\operatorname*{arg\,min}_{\alpha}
    \frac{1}{2}\sum_{i=1}^{M}
    \left\lVert C\hat x_N^{(i),*}-\tilde y_N^{(i)}\right\rVert^2
    \label{eq:pem-cost}\\
    \text{s.t.}\quad
    &\mathrm v^{(i),*}\in\operatorname*{arg\,min}_{\mathrm v^{(i)}}
    \mathcal M^{(i)}(\mathrm v^{(i)},\alpha),
    \quad i=1,\ldots,M,
    \label{eq:pem-mhe-solution}\\
    &\qquad \mathcal H_k^{w,(i)}(\mathrm v^{(i)},\alpha)=0,
    \quad k=0,\ldots,N-1,
    \label{eq:pem-mhe-dynamics}\\
    &\qquad \mathcal H_k^{\eta,(i)}(\mathrm v^{(i)},\alpha)=0,
    \quad k=1,\ldots,N-1.
    \label{eq:pem-mhe-measurement}
\end{align}
\end{subequations}
In \eqref{eq:pem-mhe-solution}--\eqref{eq:pem-mhe-measurement}, the inner MHE problem estimates each state trajectory using measurements up to time $N-1$. The terminal state $\hat x_N^{(i),*}$ yields the one-step-ahead output prediction $C\hat x_N^{(i),*}$, which is compared with the held-out measurement $\tilde y_N^{(i)}$ in \eqref{eq:pem-cost}. We minimize these prediction errors over all $M$ windows to estimate the parameter $\alpha^*$.

To solve the bi-level problem~\eqref{eq:pem-objective}, we replace the inner MHE global optimality requirement \eqref{eq:pem-mhe-solution}--\eqref{eq:pem-mhe-measurement} with its first-order KKT conditions. For compactness, let $\mathcal H(\mathrm v,\alpha)=0$ denote all stacked equality constraints in \eqref{eq:pem-mhe-dynamics}--\eqref{eq:pem-mhe-measurement}, and let $\mu$ denote the correspondingly Lagrange multiplier. 

The Lagrangian of the inner MHE problem is
\begin{equation}
    \mathcal L(\mathrm v,\mu,\alpha)
    :=\mathcal M(\mathrm v,\alpha)+\mu^\top\mathcal H(\mathrm v,\alpha),
    \label{eq:mhe-lagrangian}
\end{equation}
and the corresponding KKT conditions are
\begin{equation}
    \mathcal F(\mathrm v,\mu,\alpha)
    :=\begin{bmatrix}
        \nabla_{\mathrm v}\mathcal L(\mathrm v,\mu,\alpha)\\
        \mathcal H(\mathrm v,\alpha)
    \end{bmatrix}=0.
    \label{eq:mhe-kkt}
\end{equation}

This yields the following relaxed single-level optimization problem:
\begin{subequations}
\label{eq:pem-mhe-problem}
\begin{align}
    (\mathrm v^{(i),*},\mu^{(i),*},\alpha^*)
    &\in\operatorname*{arg\,min}_{(\mathrm v^{(i)},\mu^{(i)},\alpha)}
    \frac{1}{2}\sum_{i=1}^{M}
    \left\lVert C\hat x_N^{(i)}-\tilde y_N^{(i)}\right\rVert^2
    \label{eq:pem-mhe-objective}\\
    \text{s.t.}\quad
    &\mathcal F^{(i)}(\mathrm v,\mu,\alpha)=0,
    \quad i=1,\ldots,M.
    \label{eq:pem-mhe-kkt-constraint}
\end{align}
\end{subequations}

A key advantage of \eqref{eq:pem-mhe-problem} is parallel computation across the $M$ data windows.

\subsection{Scalable Hypergradient Computation}
\label{subsec:matrix-free-hypergradient}

To obtain the parameter estimate $\alpha^*$, we aim to solve \eqref{eq:pem-mhe-problem} using gradient-based optimization algorithms. For a given parameter $\alpha$, $\mathrm v^{(i)}$ and $\mu^{(i)}$ are regarded as implicit functions of $\alpha$ through the KKT constraint \eqref{eq:pem-mhe-kkt-constraint}. We denote the resulting terminal state estimate as $\hat x_N^{(i)}(\alpha)$. Thus, we can express the PEM objective as a function of $\alpha$:
\begin{equation}
    \mathcal J(\alpha)
    :=
    \frac{1}{2}\sum_{i=1}^{M}
    \left\lVert
    C\hat x_N^{(i)}(\alpha)-\tilde y_N^{(i)}
    \right\rVert^2,
    \label{eq:reduced-pem-objective}
\end{equation}
We refer to the total derivative $\frac{\mathrm d\mathcal J}{\mathrm d\alpha}$ as the hypergradient. Computing it requires the sensitivities of $\mathrm v^{(i)}$ and $\mu^{(i)}$ with respect to $\alpha$, obtained by implicitly differentiating \eqref{eq:pem-mhe-kkt-constraint}:
\begin{subequations}
\label{eq:kkt-implicit-differentiation}
\begin{align}
    &\frac{\mathrm d}{\mathrm d\alpha}
    \mathcal F^{(i)}(\mathrm v,\mu,\alpha)=0
    \label{eq:kkt-total-derivative}\\
    \Leftrightarrow\quad
    &\underbrace{\begin{bmatrix}
        \dfrac{\partial^2\mathcal L^{(i)}}{\partial(\mathrm v^{(i)})^2} &
        \left(\dfrac{\partial\mathcal H^{(i)}}{\partial\mathrm v^{(i)}}\right)^\top\\
        \dfrac{\partial\mathcal H^{(i)}}{\partial\mathrm v^{(i)}} & 0
    \end{bmatrix}}_{K^{(i)}}
    \begin{bmatrix}
        \dfrac{\mathrm d\mathrm v^{(i)}}{\mathrm d\alpha}\\
        \dfrac{\mathrm d\mu^{(i)}}{\mathrm d\alpha}
    \end{bmatrix}
    +
    \underbrace{\begin{bmatrix}
        \dfrac{\partial^2\mathcal L^{(i)}}{\partial\mathrm v^{(i)}\partial\alpha}\\
        \dfrac{\partial\mathcal H^{(i)}}{\partial\alpha}
    \end{bmatrix}}_{A^{(i)}}
    =0.
    \label{eq:kkt-sensitivity-system}
\end{align}
\end{subequations}
% [say something like K assumed non-singular, or what happens if singular]
For the subsequent analysis, we assume that $K^{(i)}$ is nonsingular at the MHE solution $(\mathrm v^{(i),*},\mu^{(i),*})$ for all $i$. This is a standard regularity assumption in nonlinear-program sensitivity analysis, also commonly imposed in differentiable Model Predictive Control (MPC) \citep{frey2025differentiable}, ensuring that the local MHE solution sensitivities with respect to $\alpha$ are uniquely determined.

Using this sensitivity relation in \eqref{eq:kkt-sensitivity-system}, with $r^{(i)}:=C\hat x_N^{(i)}-\tilde y_N^{(i)}$ denoting the prediction residual, the hypergradient follows from the chain rule as

\begin{align}
\label{eq:explicit-matrix-hypergradient}
\frac{\mathrm d\mathcal J}{\mathrm d\alpha}
    &=
    \sum_{i=1}^{M}
    \begin{bmatrix}
        \dfrac{\mathrm d\mathrm v^{(i)}}{\mathrm d\alpha}\\
        \dfrac{\mathrm d\mu^{(i)}}{\mathrm d\alpha}
    \end{bmatrix}^{\!\top}
    \begin{bmatrix}
        \left(\dfrac{\partial(C\hat x_N^{(i)})}{\partial\mathrm v^{(i)}}\right)^\top\\
        \left(\dfrac{\partial(C\hat x_N^{(i)})}{\partial\mu^{(i)}}\right)^\top
    \end{bmatrix}r^{(i)},
    \\
    &=
    -\sum_{i=1}^{M}
    \left(A^{(i)}\right)^\top
    \left(K^{(i)}\right)^{-1}
    \underbrace{\begin{bmatrix}
        \left(\dfrac{\partial(C\hat x_N^{(i)})}{\partial\mathrm v^{(i)}}\right)^\top r^{(i)}\\
        0
    \end{bmatrix}}_{b^{(i)}}.
    \nonumber
\end{align}

The explicit hypergradient expression in \eqref{eq:explicit-matrix-hypergradient} presents a computational challenge when the parameter dimension $n_\theta$ is large. In particular, $A^{(i)}$ contains $n_\alpha\geq n_\theta$ columns and for neural and gray-box parameterizations we typically have $n_\theta\gg1$. Thus, the explicit construction of $A^{(i)}$ requires a large amount of memory and the complexity of the associated multiplication increases with $n_\theta$. To address this parameter-scaling bottleneck, the following proposition reformulates the hypergradient computation.

\begin{prop}[Matrix-free adjoint--VJP hypergradient]
\label{prop:matrix-free-hypergradient}
\mbox{}\par\noindent
The exact hypergradient in \eqref{eq:explicit-matrix-hypergradient} can be reformulated as
\begin{equation}
    \frac{\mathrm d\mathcal J}{\mathrm d\alpha}
    =-\sum_{i=1}^{M}
    \nabla_\alpha\!\left[
        \left(\mathcal F^{(i)}(\mathrm v^{(i)},\mu^{(i)},\alpha)\right)^\top\lambda^{(i)}
    \right],
    \label{eq:matrix-free-hypergradient}
\end{equation}
where the adjoint vector $\lambda^{(i)}$ is obtained from
\begin{equation}
    K^{(i)}\lambda^{(i)} = b^{(i)}.
    \label{eq:adjoint-system}
\end{equation}

% In \eqref{eq:matrix-free-hypergradient}, $\mathrm v^{(i)}$, $\mu^{(i)}$, and $\lambda^{(i)}$ are held fixed when evaluating the gradient with respect to $\alpha$.
\end{prop}

\begin{pf}
For each window $i$, since $K^{(i)}$ is nonsingular, \eqref{eq:adjoint-system} gives $\lambda^{(i)}=(K^{(i)})^{-1}b^{(i)}$. Substitution into the  chain-rule expression \eqref{eq:explicit-matrix-hypergradient} yields
\begin{equation*}
    \frac{\mathrm d\mathcal J}{\mathrm d\alpha}
    =-\sum_{i=1}^{M}\left(A^{(i)}\right)^\top\lambda^{(i)}.
\end{equation*}
Since $A^{(i)}=\frac{\partial\mathcal F^{(i)}(\mathrm v^{(i)},\mu^{(i)},\alpha)}{\partial\alpha}$, $\mathrm v^{(i)}$, $\mu^{(i)}$ and $\lambda^{(i)}$ are fixed, we can equivalently write
\begin{equation*}
    \left(A^{(i)}\right)^\top\lambda^{(i)}
    =
    \nabla_\alpha\!\left[
        % \left(\mathcal F^{(i)}\right)^\top\lambda^{(i)}
        \left(\mathcal F^{(i)}(\mathrm v^{(i)},\mu^{(i)},\alpha)\right)^\top\lambda^{(i)}
    \right].
\end{equation*}
Therefore, \eqref{eq:matrix-free-hypergradient} follows. $\square$
\end{pf}

Our reformulation replaces the explicit matrix computation in \eqref{eq:explicit-matrix-hypergradient} with two VJP evaluations and one linear-system solve. The first VJP computes $\nabla_{\mathrm v}\mathcal L$ in $\mathcal F$, the linear-system solve \eqref{eq:adjoint-system} yields the adjoint $\lambda^{(i)}$, and the second VJP computes the gradient of the scalar $(\mathcal F^{(i)})^\top\lambda^{(i)}$ with respect to $\alpha$.

By avoiding the explicit construction and storage of the $n_\alpha$-column matrix $A^{(i)}$, the proposed formulation removes this parameter-scaling bottleneck and substantially reduces the computational and memory requirements for neural and gray-box parameterizations.

We next consider how to solve the adjoint system \eqref{eq:adjoint-system} efficiently. To analyze its computational complexity, let
\begin{equation}
s:=2n_{\mathrm x}+n_{\mathrm y}+n_w+n_\eta
\end{equation}
denote the maximum dimension of the adjoint variables associated with a single MHE stage. Treating the adjoint linear system in \eqref{eq:adjoint-system} as dense requires $\mathcal O(N^3s^3)$ time and $\mathcal O(N^2s^2)$ memory, presenting a computational challenge when a longer horizon $N$ is required to incorporate longer history. To address this horizon-scaling bottleneck, the following proposition exploits the temporal structure of the inner MHE problem.

\begin{prop}[Structure-exploiting adjoint solve]
\label{prop:structured-adjoint-solve}
\mbox{}\par\noindent
There exists a permutation matrix $P$ that reorders the adjoint as $\bar\lambda:=P\lambda$, such that $\bar K:=PKP^\top$ is block tridiagonal. With $\bar b:=Pb$, the reordered adjoint system $\bar K\bar\lambda=\bar b$ can be solved in $\mathcal O(Ns^3)$ time and in $\mathcal O(Ns^2)$ memory.
\end{prop}

\begin{pf}
Following the block structure of $\mathcal F$ in \eqref{eq:mhe-kkt}, let $\lambda^{\mathcal L}_{x_k}$, $\lambda^{\mathcal L}_{w_k}$, and $\lambda^{\mathcal L}_{\eta_k}$ denote the adjoint components corresponding to the entries of $\nabla_{\mathrm v}\mathcal L$ associated with $\hat x_k$, $\hat w_k$, and $\hat\eta_k$, respectively. Similarly, let $\lambda^{\mathcal H}_{w_k}$ and $\lambda^{\mathcal H}_{\eta_k}$ denote the adjoint components corresponding to $\mathcal H_k^w$ and $\mathcal H_k^\eta$. We then reorder these components in temporal order as
\begin{align}
    \bar\lambda_0
    &:=\operatorname{col}\!\left(
        \lambda^{\mathcal L}_{x_0},
        \lambda^{\mathcal L}_{w_0},
        \lambda^{\mathcal H}_{w_0}
    \right),\notag\\
    \bar\lambda_k
    &:=\operatorname{col}\!\left(
        \lambda^{\mathcal L}_{x_k},
        \lambda^{\mathcal L}_{w_k},
        \lambda^{\mathcal L}_{\eta_k},
        \lambda^{\mathcal H}_{w_k},
        \lambda^{\mathcal H}_{\eta_k}
    \right),
    \quad k=1,\ldots,N-1,\notag\\
    \bar\lambda_N
    &:=\lambda^{\mathcal L}_{x_N},\notag\\
    \bar\lambda
    &:=\operatorname{col}(\bar\lambda_0,\ldots,\bar\lambda_N).
    \label{eq:temporal-adjoint-ordering}
\end{align}

This reordering can be represented by a permutation matrix $P$ as $\bar\lambda=P\lambda$.
% The MHE problem \eqref{eq:inner-mhe} has the same multi-stage structure as MPC~\citep{diehl2009efficient}.
At sampling instant $k$, the objective term in \eqref{eq:mhe-objective} is stage-local, the dynamics constraint \eqref{eq:mhe-dynamics-constraint} involves only $(\hat x_k,\hat w_k,\hat x_{k+1})$, and the measurement constraint \eqref{eq:mhe-measurement-constraint} involves only $(\hat x_k,\hat\eta_k)$. Therefore, the KKT conditions \eqref{eq:mhe-kkt} associated with time $k$ depend only on variables at adjacent time stages $k-1$, $k$, and $k+1$. By the definition of $K$ in \eqref{eq:kkt-sensitivity-system}, the blocks coupling nonadjacent time stages are zero and lie outside the main block diagonal and first block off-diagonals of $\bar K$. Hence, $\bar K$ is block tridiagonal. 
The block-tridiagonal system $\bar K\bar\lambda=\bar b$ can be solved using block-factorization methods, such as Riccati recursion~\citep{diehl2009efficient}, in $\mathcal O(Ns^3)$ time and $\mathcal O(Ns^2)$ memory. $\square$
\end{pf}

This linear scaling in $N$ is comparable to that of structure-exploiting solvers used for MHE and MPC schemes.

Overall, the proposed formulation provides a scalable PEM--MHE identification approach while preserving the exact hypergradient and remaining parallelizable across the $M$ data windows. It supports high-dimensional parameterizations and improves scalability to longer estimation horizons.
\section{Experiments}
\label{sec:experiments}

In this section, we evaluate the proposed scalable extension of PEM--MHE in terms of computational efficiency, parameter recovery, and prediction accuracy. We consider the identification of nonlinear racing-vehicle dynamics under partial state observation, which provides a challenging benchmark due to the strongly nonlinear tire-force dynamics and the unmeasured lateral velocity. We first evaluate the computational scaling of the proposed hypergradient reformulation (Sec.~\ref{subsec:experiments_compute}). We then evaluate parameter recovery and prediction accuracy using simulated vehicle data provided by \cite{chrosniak2024deep} and generated with BayesRace~\citep{jain2020bayesrace}  (Sec.~\ref{subsec:simulation-results}). Lastly, we validate the approach on real-world data collected from an Indy Autonomous Challenge racecar~\citep{kulkarni2023racecar} and processed by \cite{chrosniak2024deep}  (Sec.~\ref{subsec:realworld-results}).

\textit{Implementation details:} 
The MHE problem~\eqref{eq:inner-mhe} is formulated in CasADi, with neural models integrated via L4CasADi~\citep{salzmann2023neural,salzmann2024l4casadi}, and solved using acados~\citep{verschueren2022acados}. The adjoint--VJP hypergradient~\eqref{eq:matrix-free-hypergradient} is implemented using PyTorch~\citep{paszke2019pytorch}. Bounded parameters are reparameterized in terms of unconstrained parameters using softplus or sigmoid mappings for one-side and two-sided bounds, respectively. All unconstrained parameters are optimized using Adam~\citep{kingma2014adam}. 
The proposed implementation and all numerical values are available online: \\{\url{https://github.com/YiboFrankZhou/Scalable-PEM-MHE}}.

\subsection{Setup}
\label{subsec:tire-dynamics}

We define the vehicle state and input as
\begin{equation}
\label{eq:vehicle-state-input}
    x:=\begin{bmatrix}v_{\mathrm{x}}&v_{\mathrm{y}}&\dot\varphi\end{bmatrix}^{\top},
    \quad
    u:=\begin{bmatrix}\delta&\gamma\end{bmatrix}^{\top},
\end{equation}
where $v_{\mathrm{x}}\in\mathbb{R}$ and $v_{\mathrm{y}}\in\mathbb{R}$ are the body-frame longitudinal and lateral velocities, $\dot\varphi\in\mathbb{R}$ is the yaw rate, $\delta\in\mathbb{R}$ is the steering angle, and $\gamma\in\mathbb{R}$ is the throttle command, see also the illustration in Fig.~\ref{fig:tire-convention}.
%
%Among the three states, $v_{\mathrm{x}}$ and $\dot\varphi$ are measured, whereas $v_{\mathrm{y}}$ is unmeasured. The measurement model is therefore given by
We have noisy measurements of the longitudinal velocity and the yaw rate:
\begin{equation}
    y_k=\begin{bmatrix} v_{\mathrm{x}} \\ \dot\varphi  \end{bmatrix}+\eta_k.
    \label{eq:vehicle-output}
\end{equation}
The dynamics are given by
\begin{subequations}
\begin{align}
    \dot x_{\tau}&=f(x_{\tau},u_{\tau},z_{\tau};\theta)+w_{\tau}
    \label{eq:vehicle-continuous-dynamics}\\
    z_{\tau}&=h(x_{\tau},u_{\tau};\theta).
    \label{eq:vehicle-component-model}
\end{align}
\end{subequations}
Here, $f$ denotes a dynamic single-track vehicle model~\citep{chrosniak2024deep}, with $z_\tau$ collecting the tire-frame longitudinal and lateral forces at the equivalent front and rear axles:
\begin{equation}
    z_\tau:=
    \begin{bmatrix}
        F_{\mathrm{x},\tau}^{\mathrm f}&
        F_{\mathrm{y},\tau}^{\mathrm f}&
        F_{\mathrm{x},\tau}^{\mathrm r}&
        F_{\mathrm{y},\tau}^{\mathrm r}
    \end{bmatrix}^{\top}\in\mathbb{R}^4.
    \label{eq:vehicle-force-vector}
\end{equation}
where $\{\mathrm f,\mathrm r\}$ denotes the front and rear axles.

We consider three model classes for $h$:
\begin{subequations}
\label{eq:vehicle-force-models}
\begin{itemize}
    \item \textbf{Physics-based model:}
    \begin{equation}
        h_{\mathrm{phys}}(x,u;\theta_{\mathrm{phys}})
        :=
        \begin{bmatrix}
            0&
            F_{\mathrm{Magic}}^{\mathrm f}&
            F_{\mathrm{prop}}&
            F_{\mathrm{Magic}}^{\mathrm r}
        \end{bmatrix}^{\top};
        \label{eq:vehicle-physical-forces}
    \end{equation}

    \item \textbf{Black-box neural-network model:}
    \begin{equation}
        h_{\mathrm{NN}}(x,u;\theta_{\mathrm{NN}})
        :=
        \mathrm{NN}(x,u;\theta_{\mathrm{NN}});
        \label{eq:vehicle-nn-forces}
    \end{equation}

    \item \textbf{Physics--neural gray-box model:}
    \begin{equation}
        h_{\mathrm{gray}}(x,u;\theta_{\mathrm{gray}})
        :=
        h_{\mathrm{phys}}(x,u;\theta_{\mathrm{phys}})
        +
        h_{\mathrm{NN}}(x,u;\theta_{\mathrm{NN}}).
        \label{eq:vehicle-gray-forces}
    \end{equation}
\end{itemize}
\end{subequations}

\begin{comment}
We consider three model classes for $h$: a physics-based model $h_{\mathrm{phys}}$, a black-box neural-network model $h_{\mathrm{NN}}$, and a physics-neural gray-box model $h_{\mathrm{gray}}$ combining the two:
\begin{subequations}
\label{eq:vehicle-force-models}
\begin{align}
h_{\mathrm{phys}}(x,u;\theta_{\mathrm{phys}})&:=\begin{bmatrix}0&F_{\mathrm{Magic}}^{\mathrm f}&F_{\mathrm{prop}}&F_{\mathrm{Magic}}^{\mathrm r}\end{bmatrix}^{\top},
\label{eq:vehicle-physical-forces}\\
h_{\mathrm{NN}}(x,u;\theta_{\mathrm{NN}})&:=\mathrm{NN}(x,u;\theta_{\mathrm{NN}}),
\label{eq:vehicle-nn-forces}\\
h_{\mathrm{gray}}(x,u;\theta_{\mathrm{gray}})&:=h_{\mathrm{phys}}(x,u;\theta_{\mathrm{phys}})+h_{\mathrm{NN}}(x,u;\theta_{\mathrm{NN}}).
\label{eq:vehicle-gray-forces}
\end{align}
\end{subequations}
\end{comment}

The physical model $h_{\mathrm{phys}}$ uses the Pacejka Magic Formula tire model~\citep{pacejka1992magic} for the front and rear lateral forces $F_{\mathrm{Magic}}^{\mathrm f}$ and $F_{\mathrm{Magic}}^{\mathrm r}$, and the drivetrain model~\citep{jain2020bayesrace} for the rear longitudinal force $F_{\mathrm{prop}}$. The resulting unknown physical parameters are collected in $\theta_{\mathrm{phys}}\in\mathbb{R}^{16}$.

\begin{comment}
as
\begin{equation}
\begin{aligned}
\theta_{\mathrm{phys}}:=\left[\right.&B^{\mathrm f},C^{\mathrm f},D^{\mathrm f},E^{\mathrm f},S_{\mathrm H}^{\mathrm f},S_{\mathrm V}^{\mathrm f},\\
&B^{\mathrm r},C^{\mathrm r},D^{\mathrm r},E^{\mathrm r},S_{\mathrm H}^{\mathrm r},S_{\mathrm V}^{\mathrm r},\\
&C_{\mathrm m1},C_{\mathrm m2},C_{\mathrm r0},C_{\mathrm r2}\left.\right]^{\top}
\in\mathbb{R}^{16}.
\end{aligned}
\label{eq:vehicle-physical-parameters}
\end{equation}
\end{comment}

The neural-network model $h_{\mathrm{NN}}$ uses a multilayer perceptron (MLP) with two hidden layers and $\tanh$ activation functions, parameterized by $\theta_{\mathrm{NN}}\in\mathbb{R}^{n_{\mathrm{NN}}}$.

Finally, the gray-box model $h_{\mathrm{gray}}$ combines the above physical and neural models, with the parameters collected as
\begin{equation}
\theta_{\mathrm{gray}}:=\begin{bmatrix}\theta_{\mathrm{phys}}^{\top}&\theta_{\mathrm{NN}}^{\top}\end{bmatrix}^{\top}
\in\mathbb{R}^{16+n_{\mathrm{NN}}}.
\label{eq:vehicle-gray-parameters}
\end{equation}

% [todo: remark: For the graybox nn, last layer weights are initially set to be 0.]

\begin{figure}[t]
    \centering
    \includegraphics[width=0.65\columnwidth]{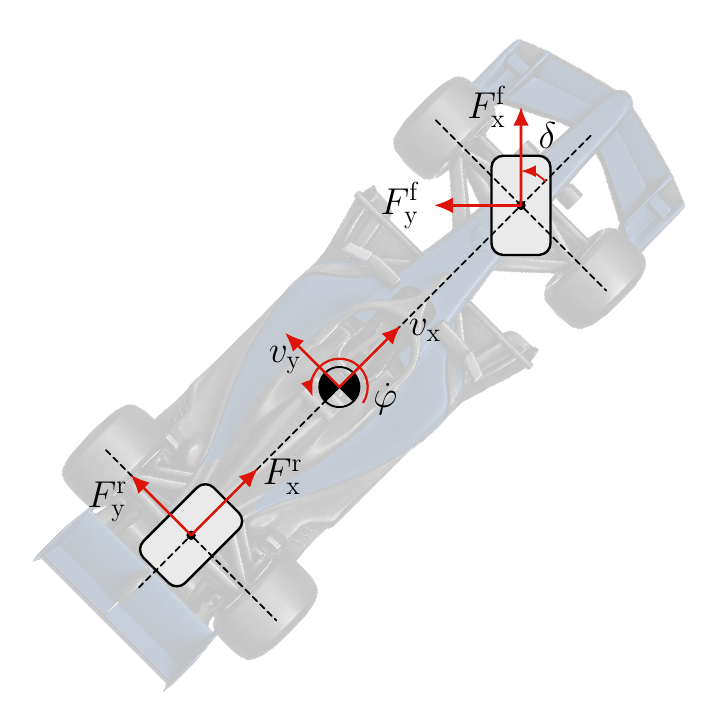}
    \caption{Vehicle geometry and modeling conventions. Schematic adapted from \cite{chrosniak2024deep}.}
    \label{fig:tire-convention}
\end{figure}

\subsection{Computational analysis}
\label{subsec:experiments_compute}
This section compares the proposed adjoint--VJP hypergradient formulation in \eqref{eq:matrix-free-hypergradient} with the explicit matrix-based formulation following \citet{simpson2024parallelizable} in \eqref{eq:explicit-matrix-hypergradient}, evaluating the computational speedup as the model parameter dimension $n_{\theta}$ increases and verifying their numerical equivalence. 
To this end, we use the neural-network model $h_{\mathrm{NN}}$ and vary the neural-network parameter dimension $n_{\mathrm{NN}}$. The final layer of the neural network is initialized to zero to avoid random initial outputs and facilitate reliable convergence of the MHE solver at initialization. We measure only the hypergradient evaluation time, with $\alpha$ and the corresponding KKT solution $(\mathrm v,\mu)$ fixed. The MHE horizon is fixed at $N=40$.

\begin{figure}[t]
    \centering
    \caption{Median hypergradient computation time over 100 MHE windows versus NN parameter count $n_{\mathrm NN}$.}
    \label{fig:hypergradient_runtime_scaling_param}
    \includegraphics[width=0.95\columnwidth]{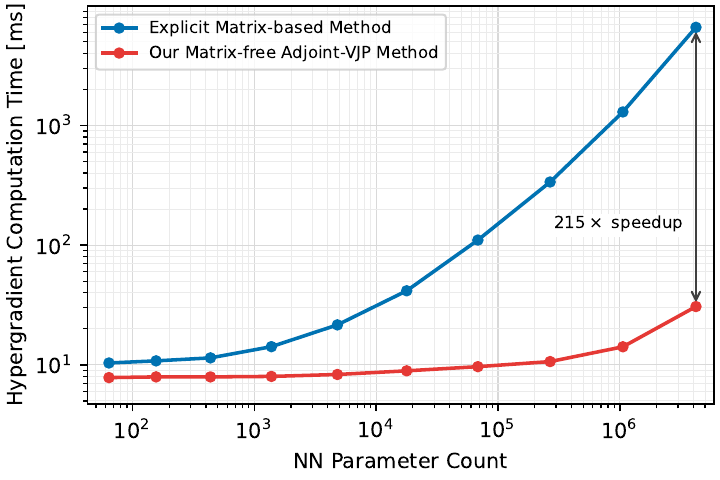}
\end{figure}

\begin{comment}
\begin{figure}[H]
    \centering
    \caption{Median adjoint computation time versus MHE horizon.}
    \label{fig:hypergradient_runtime_scaling_horizon}
    \includegraphics[width=0.95\columnwidth]{figures/Ex1_2_adjoint_horizon_scaling.pdf}
\end{figure}
\end{comment}

The results are shown in Fig.~\ref{fig:hypergradient_runtime_scaling_param}. The proposed adjoint--VJP formulation achieves increasing speedups as $n_{\mathrm{NN}}$ grows, reaching $215\times$ at 4.22 million parameters. We also verified that the resulting hypergradients are virtually indistinguishable, with a maximum relative $2$-norm error of $9.80\times10^{-7}$.

\subsection{Simulation: Parameter Recovery and Prediction}
\label{subsec:simulation-results}

\newcommand{\tmin}[1]{\cellcolor{ForestGreen!20}\bfseries\boldmath $#1$}
\newcommand{\tmax}[1]{\cellcolor{Red!18}$#1$}

\begin{table*}[!t]
    \centering
    \caption{Identification results on vehicle data simulated with BayesRace, comparing one-step prediction error, physical-parameter recovery, and training time. Best and worst error values are highlighted in green and red, respectively.}
    \label{tab:simulation}
    \small
    \setlength{\tabcolsep}{0pt}
    \begin{tabular}{@{}>{\raggedright\arraybackslash}p{.19\textwidth}@{\hspace{.0075\textwidth}}
                    >{\centering\arraybackslash}p{.115\textwidth}@{}>{\centering\arraybackslash}p{.115\textwidth}
                    @{\hspace{.0075\textwidth}}>{\centering\arraybackslash}p{.19\textwidth}
                    @{\hspace{.0075\textwidth}}>{\centering\arraybackslash}p{.105\textwidth}@{}}
        \toprule
        \multirow{2}{*}[-.8ex]{Method}
        & \multicolumn{2}{c}{One-Step Prediction Error}
        & \multirow{2}{*}[-1ex]{\shortstack{Average Relative\\ Physical-Parameter Error}}
        & \multirow{2}{*}[-1ex]{\shortstack{Training\\Time}} \\
        \cmidrule(l{4pt}r{.0075\textwidth}){2-3}
        & Train & Validation & & \\
        \midrule
        Ground True Parameters & $6.27\times10^{-14}$ & $9.71\times10^{-14}$ & $0$ & -- \\
        \midrule
        SUBNET & \tmax{2.85\times10^{-2}} & \tmax{1.66\times10^{-1}} & -- & 2 min \\
        Joint-MHE & $1.37\times10^{-2}$ & $4.42\times10^{-3}$ & $69.05\%$ & 7 min \\
        EKF$\rightarrow$PEM & $1.33\times10^{-4}$ & $6.49\times10^{-5}$ & $22.99\%$ & 1 min \\
        EKF$\rightarrow$PEM (adaptive) & $1.60\times10^{-4}$ & $9.52\times10^{-5}$ & \tmax{161.93\%} & 26 min \\
        PEM-MHE (physics-only) & \tmin{7.15\times10^{-7}} & \tmin{7.93\times10^{-7}} & \tmin{7.50\%} & 4 h 41 min \\
        % PEM-MHE (physics--NN) & $7.43\times10^{-7}$ & $7.88\times10^{-7}$ & $4.13$ & 15 h 59 min \\
        \bottomrule
    \end{tabular}
\end{table*}

This section evaluates the parameter recovery and prediction accuracy of the proposed PEM--MHE framework against other identification methods.

In this simulation setting, the data-generating model structure is known and matches that of the physical model $h_{\mathrm{phys}}$ used for identification. Accordingly, PEM--MHE and all structured-model baselines use the same physical model $h_{\mathrm{phys}}$ in this experiment, without the need for a gray-box extension. The ground-truth parameter values are also known, allowing direct evaluation of the estimated physical parameters' accuracy.

The following algorithms are compared with the proposed PEM--MHE framework:
\begin{itemize}
    \item \textbf{SUBNET}~\citep{beintema2023deep} is a black-box neural identification method that propagates a learned latent state. Hence, the learned latent states do in general not align with the physical state coordinates. For comparison with the other methods, we set the latent-state dimension to three, the true physical dimension.
    \item \textbf{Joint-MHE}~\citep{kuhl2011real} jointly estimates the states and model parameters by minimizing the MHE objective.
    \item \textbf{EKF$\rightarrow$PEM} follows a two-stage identification strategy~\citep[Ch.~18]{tedrake2023underactuated}, first estimating the unmeasured states using an EKF based on an initial model, and then identifying the model parameters by minimizing the one-step prediction error while treating the state estimates as ground truth.
    \item \textbf{EKF$\rightarrow$PEM (adaptive)}~\citep{chrosniak2024deep} similarly treats an EKF state estimate as ground truth and minimizes the one-step prediction error. Instead of estimating a constant parameter vector, an NN is trained to map a window of past state estimates and inputs to a time-varying parameter estimate.
    %uses the same EKF preprocessing and terminal prediction objective. Instead of a constant parameter vector, an encoder maps a window of past state estimates and inputs to time-varying parameters. These parameters are used $h_{\mathrm{phys}}$ for the terminal prediction.
\end{itemize}

For a fair and consistent comparison, all prediction-based methods use the first $N-1$ steps of each data window to predict the terminal output $\hat y_N$, with the one-step prediction error defined as its mean squared error against the measurement $\tilde y_N$ over all windows. The relative error for each physical parameter $\theta_{\mathrm{phys},j}$ is computed as
\begin{equation}
    \frac{|\hat\theta_{\mathrm{phys},j}-\theta_{\mathrm{phys},j}^*|}{|\theta_{\mathrm{phys},j}^*|},
\end{equation}
where $\theta_{\mathrm{phys},j}^*$ denotes the ground-truth value. We report the average error over all 16 physical parameters. For EKF$\rightarrow$PEM (adaptive), the time-varying parameter estimates are first averaged over all validation windows.

\textit{Initialization:} SUBNET requires no initial physical parameters. Joint-MHE and the proposed PEM--MHE use the same physical-parameter initialization, whose average relative physical-parameter error is $67.95\%$. The EKF-based methods rely on the initial model for state estimation and are thus more fragile to large parameter error. In order to provide meaningful results for these methods, we randomly sampled $10$ parameter vectors within $\pm5\%$ error and only report the result with the initialization where the EKF-based method achieves the smallest validation error. 

Table~\ref{tab:simulation} summarizes the results. The proposed PEM--MHE achieves the lowest training and validation errors, with a validation error over $\mathbf{80}$ \textbf{times smaller} than any other method. PEM--MHE further achieves the lowest physical-parameter error.

\begin{table*}[!t]
    \centering
    \caption{Identification results on the IAC real-world vehicle dataset, comparing one-step prediction error and training time. Best and worst error values are highlighted in green and red, respectively.}
    \label{tab:real-world}
    \small
    \setlength{\tabcolsep}{0pt}
    \begin{tabular}{@{}>{\raggedright\arraybackslash}p{.19\textwidth}@{\hspace{.0075\textwidth}}
                    >{\centering\arraybackslash}p{.115\textwidth}@{}>{\centering\arraybackslash}p{.115\textwidth}
                    @{\hspace{.0075\textwidth}}>{\centering\arraybackslash}p{.105\textwidth}@{}}
        \toprule
        \multirow{2}{*}[-.8ex]{Method}
        & \multicolumn{2}{c}{One-Step Prediction Error}
        & \multirow{2}{*}[-1ex]{\shortstack{Training\\Time}} \\
        \cmidrule(l{4pt}r{.0075\textwidth}){2-3}
        & Train & Validation & \\
        \midrule
        SUBNET & $3.16\times10^{-4}$ & \tmax{3.25\times10^{-3}} & 3 min \\
        Joint-MHE & \tmax{4.70\times10^{-4}} & $4.47\times10^{-4}$ & 4 min \\
        EKF$\rightarrow$PEM & $4.37\times10^{-4}$ & $4.20\times10^{-4}$ & 1 min \\
        EKF$\rightarrow$PEM (adaptive) & \tmin{2.20\times10^{-4}} & \tmin{2.62\times10^{-4}} & 4 min \\
        PEM-MHE (physics-only) & $4.24\times10^{-4}$ & $4.10\times10^{-4}$ & 2 h 50 min \\
        PEM-MHE (gray-box) & $2.92\times10^{-4}$ & $2.84\times10^{-4}$ & 23 h 40 min \\
        \bottomrule
    \end{tabular}
\end{table*}

The comparison with Joint-MHE and EKF$\rightarrow$PEM highlights the importance of coupling state estimation with a prediction-error identification objective under partial observation. In Joint-MHE, jointly estimating unknown states and parameters with a single MHE objective allows mutual compensation, so a low estimation cost does not necessarily indicate recovery of the underlying dynamics. In EKF$\rightarrow$PEM, errors in the decoupled state estimates cannot be corrected because the updated model parameters do not refine the state estimates. In contrast, PEM--MHE uses the PEM and MHE objectives for parameter identification and state estimation, respectively, while coupling the two within a unified optimization problem.

The comparison with SUBNET highlights the benefit of prescribed physical structure. Such constraints on the latent coordinates enable PEM--MHE to achieve more accurate estimation while preserving interpretability. Increasing the latent-state dimension of SUBNET may alleviate this limitation.

These improvements come at the cost of increased computational demand due to the nested optimization structure, where each PEM update requires solving and differentiating through the inner nonlinear MHE problem.

The adaptive EKF$\rightarrow$PEM provides no benefit in this simulation, as the true parameters are constant.

\FloatBarrier

\subsection{Real-World Data Validation}
\label{subsec:realworld-results}

We next extend the validation to real-world vehicle data, where unmodeled and time-varying effects provide a more challenging setting.

The baseline methods are the same as those in Section~\ref{subsec:simulation-results}, while PEM--MHE is evaluated with both the physics-based model $h_{\mathrm{phys}}$ and the gray-box model $h_{\mathrm{gray}}$. For the gray-box model, the neural residual uses two hidden layers of width $32$, resulting in $n_\theta=1{,}396$. Its final layer is initialized to zero.
All methods involving physical parameters are initialized with parameter estimate obtained by \citet{chrosniak2024deep} on the same dataset.
The results are shown in Table~\ref{tab:real-world}.

The value of prescribed physical structure is  evident in the substantial train--validation prediction-error gap of SUBNET, whereas the structured models exhibit substantially more consistent generalization.

Building on this prescribed physical structure, the neural residual provides additional flexibility for capturing aspects of the real-world dynamics that are not represented by the physical model, reducing both training and validation prediction errors by approximately $31\%$.

However, despite this improvement, the gray-box PEM--MHE is slightly outperformed by the EKF$\rightarrow$PEM (adaptive) on this dataset. This suggests that the real-world data may be better captured by time-varying parameters, while the proposed PEM--MHE formulation estimates time-invariant model parameters.
The non-adaptive EKF$\rightarrow$PEM exhibits a substantially larger prediction errors than its adaptive counterpart, supporting the hypothesis that the difference is due to the time-varying nature of the dynamics.

%Both methods introduce neural flexibility: the gray-box PEM--MHE directly augments the standard state-space model with a neural residual, whereas the adaptive PEM-EKF uses a neural encoder to infer time-varying effective parameters from historical information. 
%However, because state propagation in PEM--MHE depends only on the instantaneous predefined physical state, effects that are not captured in these states may be missed, while the adaptive encoder can exploit temporal context to account for them. This may explain the remaining performance gap.
% In particular, it seems that this data is better explained by time-varying parameters, which the current PEM-MHE formulation cannot handle.The fact that this is due to the time-varying nature of the parameters can also be seen by the fact that the base EKF$\rightarrow$PEM version has a significantly larger model error. 

% Additionally, we note that the PEM--MHE variants exhibit relatively large latent-state estimation errors compared with the baselines. At initialization, the $v_{\mathrm y}$ RMSE of PEM--MHE is $0.173\,\mathrm{m\,s^{-1}}$, lower than those of all baselines. The subsequent increase in latent-state error despite improved prediction accuracy suggests that, latent-state estimation may also absorb part of the mismatch, allowing the PEM objective to trade latent-state accuracy for improved measured-output prediction.

\section{Conclusion}
\label{sec:conclusion}

In this work, we investigated nonlinear system identification from partial measurements, which couples parameter learning with state estimation. The existing PEM--MHE framework becomes computationally intractable for models with a large number of trainable parameters, such as neural networks.
To address this challenge, we proposed a scalable reformulation
that avoids the explicit construction of large sensitivity matrices and exploits temporal sparsity for efficient hypergradient computation.
Finally, we demonstrated substantial computational speedups through numerical comparisons. This improved parameter scalability enables the use of neural residuals to improve prediction accuracy on real-world vehicle data. We also showed that the retained physical structure facilitates accurate physical-parameter recovery in simulation while preserving interpretability.
While the proposed approach improves computational scalability,
open questions remain regarding sufficient conditions for asymptotic consistency, time-varying model parameters, and further reductions in computational demand.

\section*{Declaration of Generative AI and AI-Assisted Technologies in the Manuscript Preparation Process}

During the preparation of this work, the authors used ChatGPT (OpenAI) to improve the language and readability of the manuscript. After using this tool, the authors reviewed and edited the content as needed and take full responsibility for the content of the publication. 
Codex (OpenAI) and Claude Code (Anthropic) were used to assist with writing and modifying the research code, which was subsequently reviewed by the authors.

\bibliography{references}             

\end{document}